\documentclass[preprint]{ptephy_v2}

\usepackage{hyperref}

\newtheorem{theorem}{Theorem}

\begin{document}

\title{Description of KPZ interface growth by stochastic Loewner evolution}


\author{Yusuke Kosaka Shibasaki}
\affil{Institute of Natural Sciences, College of Humanities and Sciences, Nihon University, Setagaya, Tokyo, 156-8550, Japan}
\affil{College of Art, Nihon University, Nerima, Tokyo, 176-8525, Japan \email{shibasaki.yusuke@nihon-u.ac.jp}}





\begin{abstract}%
In this study, we investigate the relationship between the one-dimensional (1D) Kardar-Parisi-Zhang (KPZ) equation and the stochastic Loewner equation (SLE), which is a one parameter family of the conformal mappings involving stochasticity. The author shows the correspondence between 1D KPZ equation with height function $h(x,t)=(3t^2x+x^3)/6t$ and Loewner equation driven by a nonlinear stochastic process, wherein the 1D dynamics of interface growth is characterized by Loewner entropy $S_{Loew}\simeq-\ln{t/\kappa}$. These results were numerically verified with discussions in relation to the universality in non-equilibrium statistical physics. 
\end{abstract}

\begin{description}
\item [\textbf{IMPORTANT NOTE}: ] This preprint is \textbf{NOT PEER-REVIEWED} article. However, I would like to share the present result for the purpose of the communications and discussions among the Experts, Educators, and Students (4/4/2026). \\
Yusuke K. Shibasaki, Nihon University.   
\end{description}

\maketitle

For the development of theoretical approaches to describe the non-equilibrium phenomena in the real-world systems, some nonlinear equations have been investigated. In 1986, Kardar, Parisi, and Zhang (KPZ) proposed an equation describing the growing interface in the pattern formation of active zone of the matters, particularly for examining underlying physical mechanisms using the analytical approach [1]. Since its suggestion, the KPZ equation has been well-investigated as one of the fundamental equations in the non-equilibrium statistical physics. The recent theoretical and experimental developments coined the term “KPZ universality class”, and various non-equilibrium phenomena are found to belong to this class [2-10]. The KPZ equation is described as a nonlinear partial differential equation with stochastic noise term, in which the analytical treatment relies on the stochastic differential equation (SDE); however, the derivation of the exact solution of this equation is known to be a difficult issue although some mathematical studies have suggested effective approaches to this problem [11-18]. 
In 2000s, on the contrary,  the model using conformal dynamics of the interface growth has been discussed by Feigenbaum, Procaccia, and Davidovich et al [19,20]. in relation to Saffman-Taylor instability [21]. Despite this conformal approach as well intends to describe the growing interface of active matters, e.g., viscous fluid, or the aggregation of the colloids, the connection between the KPZ and conformal dynamical approaches is still lacking. One of the representative theoretical tools that possesses both characteristics of conformal dynamics and SDE is stochastic Loewner evolution (SLE), which is suggested by Schramm in 2000 [22]. The standard SLE is comprised of Loewner differential equation (LDE) [23,24] with Wiener process, originally studied for the purpose of describing the 2-dimensional (2D) conformally invariant random curves. It is worth noting that some studies suggested its relation to growth process in the non-equilibrium systems, e.g., Laplacian growth [25] and diffusion-limited aggregation (DLA) [26,27], supporting the evidence that the theory of SLE helps to understand the interface growth problems in a more general context. 
For the above motivation, in this study, the author shall propose a possible description of the KPZ equation derived from Loewner equation. As is mentioned above, the Loewner differential equation expresses the 2D curve growth in the upper half of the complex plane. It takes a form of the time-differential equation the conformal map, which is driven by the real-valued one-dimensional driving process called Loewner driving function. By choosing the specific stochastic process as the Loewner driving function, we demonstrate the connection between KPZ and Loewner equation. This article is organized as the following parts. In Secs. 2.1 and 2.2, we introduce the original form of KPZ equation and the modified SLE used in this study, respectively. In Sec. 3.1, the derivation of the 1D KPZ equation from modified SLE is performed. In Sec. 3.2, the Loewner entropy corresponding to the dynamics of KPZ equation is analytically calculated. In Sec. 4, the numerical simulation to verify the present results is performed. In Sec. 5 and 6, the discussions and conclusion of this study are remarked. 

\section{Model}
\subsection{One-dimensional KPZ equation}
As a model of the interface growth, we consider the KPZ equation with one-dimensional (1D) spatial coordinate described as follows [1, 9]:
\begin{equation}
\frac{\partial h\left(x,t\right)}{\partial t}=\upsilon\frac{\partial^2h\left(x,t\right)}{\partial x^2}+\frac{\lambda}{2}\left(\frac{\partial h\left(x,t\right)}{\partial x}\right)^2+\sqrt\kappa\eta(t)
\end{equation} 
Here, $h\left(x,t\right)$ denotes the height of the interface at the position $x$ and time $t$. The first term in the right-hand side (r.h.s.) is the relaxation induced by the surface tension $\upsilon$. The second term in the r.h.s. is the nonlinear term appearing in the interface growth dynamics, whose strength is parameterized by $\lambda$. The term $\eta(t)$ is the white Gaussian noise with variance 1 and mean 0. The strength of noise term $\eta(t)$ is expressed by the constant diffusivity parameter $\kappa$. Because the KPZ equation in Eq. (1) takes a form of nonlinear Langevin equation, the derivation of its analytical solution is technically difficult; however, the scaling derived from Eq. (1) is found in various non-equilibrium systems, and it is called KPZ universality class. Let us consider the width of the surface defined as the following: [1,2,4]
\begin{equation}
W(L,t):=\sqrt{\langle h\left(L,t\right)^2 \rangle - {\langle h\left(L,t\right) \rangle}^2}. 
\end{equation}
Here, the brackets denote the ensemble averages and $L$ denotes the system size. For the KPZ equation, using the function $f$ and crossover time $t^\ast$, $W(L,t)$ obeys the following scaling relations: 
\begin{equation}
W\left(L,t\right) \sim t^\beta f\left(Lt^{\left(-1/z\right)}\right)\sim
\begin{cases} 
L^\alpha\   &\text {for  $t>t^\ast$}\\
t^\beta\  &\text{for  $t<t^\ast$}.
\end{cases}
\end{equation}
Particularly, the set of the scaling exponents $\alpha=1/2$,\ $\beta=1/3$, and $z=3/2$ in Eq. (3) corresponds to the KPZ universality class.

\subsection{Modified SLE for interface growth} 
To investigate the conformal description of the KPZ equation in Eq. (1), we here introduce the (chordal) Loewner differential equation expressed as follows. Let $\mathbb{H}$ be the upper half-plane, and $\gamma_{[0,s]}$ denote the simple curve starting from the origin O. By considering the conformal map $g_s(z)$, which maps the region $\mathbb{H}\setminus\gamma_{[0,s]}$ to $\mathbb{H}$. The following Loewner equation expresses the evolution of the conformal map $g_s(z)$. 
\begin{equation}
\frac{\partial g_s(z)}{\partial s}=\frac{2}{g_s(z)-U_s},\ \ \ \ \ \ \ \ g_0(z)=z\in\mathbb{H}.
\end{equation}
Here, $U_s$ is a real-valued 1D function called Loewner driving function. The partial differential equation in Eq. (4) is parameterized by $s$, which is the capacity parameter of the conformal map, and it is often called Loewner time. For the driving function, we can choose arbitrary 1D process to observe the corresponding curve growth on $\mathbb{H}$. For this study, we construct a curve coordinate-dependent stochastic process as the Loewner driving function expressed by following differential equation:
\begin{equation}
\frac{dU_s}{ds}=-\frac{x^2+y^2}{2y}-\sqrt\kappa\frac{2y}{x^2+y^2}\frac{dB_s}{ds}.\ \ 
\end{equation}
Here, $x$ and $y$ are variables whose stochastic behavior is same as those of $\rm{Re}\gamma_s$ and $\rm{Im}\gamma_s$, respectively. The term $B_s$ denotes the standard Brownian motion. The mathematical formulation using the backward Loewner evolution [28-30] shows that these variables evolve obeying the following nonlinear Langevin equation:
\begin{align}
\frac{dx}{ds}&=-\frac{2x}{x^2+y^2}+\frac{x^2+y^2}{2y}+\sqrt\kappa\frac{2y}{x^2+y^2}\frac{dB_s}{ds},\\ 
\frac{dy}{ds}&=\frac{2y}{x^2+y^2}.\ 
\end{align}
For the initial condition we choose $x(0)=0$ and $y(0)=\varepsilon$, where $\varepsilon$ is an infinitesimal constant. In the following section, by using coordinate transformation, we demonstrate the intrinsic equivalence of KPZ equation in Eq. (1) and nonlinear Langevin equation in Eqs. (6) and (7) with some restriction on $h\left(x,t\right)$. In addition, it should be noted that the stochastic calculus involving the noise term $\eta(t)$ is based on the It\^{o} interpretation \footnote{In the physical sense, the It\^{o} interpretation has a characteristic that there is no information flow from the future (See, Sec. 3.3 in Ref. [31]). Thus, the It\^{o} interpretation is reasonable for the deterministic modeling such as biological systems [31].}. 

\section{Results}
\subsection{Derivation of KPZ equation from modified SLE}
Combining Eqs. (6) and (7), by using the transformation $y\rightarrow t$ [32,33,34] \footnote{This time coordinate change was proposed in the physics literatures [32,33] although its mathematical meaning should be investigated in relation to the natural parametrization of SLE [34] for the further studies.}, $dx/dt$ is obtained as:
\begin{equation}
\frac{dx}{dt}=-\frac{x}{t}+\left(\frac{x^2+t^2}{2t}\right)^2+\sqrt\kappa\frac{dB_s}{ds}.\ 
\end{equation}
We impose the range of the time interval as $t\in[0,1]$. This equation corresponds to the Langevin equation with time-dependent drift term. Let us define:
\begin{equation}
\gamma\left(x,\ t\right):=\frac{1}{t}-\frac{1}{4}\frac{x^3}{t^2}-\frac{1}{2}x.
\end{equation}
Using Eq. (9) and replacing $\frac{dB_s}{ds}=\eta(t)$, Eq. (8) is rewritten as [35,36]:
\begin{equation}
\frac{dx}{dt}=-\gamma\left(x,\ t\right)x+\sqrt\kappa \eta(t).
\end{equation}
It has been shown that the formal solution of the above Langevin equation is expressed as the following:
\begin{equation}
x\left(t\right)=x\left(0\right)F\left(t\right)+\sqrt\kappa F\left(t\right)\int_{t_0}^{t}{\frac{\eta(t^\prime)\ }{F\left(t^\prime\right)}dt^\prime},\ \ 
\end{equation}
where
\begin{equation}
F\left(t\right):=\exp{\left[-\int_{0}^{t}\gamma\left(x,\ t^\prime\right)dt^\prime\right]}=\frac{1}{t}\exp{\left(\frac{xt}{2}-\frac{x^3}{4t}\right)}.
\end{equation}
For the subsequent discussion, we shall derive the scaling of the variance of the variable $x\left(t\right)$. Assuming $x\left(0\right)=0$, the variance $\langle x(t)^2\rangle$ is calculated as the following. 
\begin{equation}
\langle x(t)^2\rangle = \kappa{F\left(t\right)}^2\int_{t_0}^{t}{\frac{1}{{F(t^\prime)}^2}dt^\prime}.\ 
\end{equation}
From Eq. (12), the terms ${F\left(t\right)}^2$ and $1/{F\left(t\right)}^2$ are calculated as follows:
\begin{equation}
{F\left(t\right)}^2=\frac{1}{t^2}\exp{\left(xt-\frac{x^3}{2t}\right)},\\ 
\frac{1}{{F(t^\prime)}^2}=t^2\exp{\left(\frac{x^3}{2t^\prime}-xt^\prime\right)}.\ 
\end{equation}
Integrating by parts, we observe that
\begin{equation}
\begin{split}
\int_{t_0}^{t}{\frac{1}{{F\left(t^\prime\right)}^2}dt^\prime}&=\int_{t_0}^{t}{{t^\prime}^2\exp{\left(\frac{x^3}{2t^\prime}-xt^\prime\right)dt^\prime}}\\
&=\left[\frac{1}{3}{t^\prime}^3\exp{\left(\frac{x^3}{2t^\prime}-xt^\prime\right)}\right]_{t_0}^t-\int_{t_0}^{t}{\frac{1}{3}{t^\prime}^3\exp{\left(\frac{x^3}{2t^\prime}-xt^\prime\right)\left(-\frac{x^3}{2{t^\prime}^2}-x\right)}}dt^\prime\\
&=\left[\frac{1}{3}{t^\prime}^3\exp{\left(\frac{x^3}{2t^\prime}-xt^\prime\right)}\right]_{t_0}^t+\int_{t_0}^{t}{\frac{x}{3}{t^\prime}^3\exp{\left(\frac{x^3}{2t^\prime}-xt^\prime\right)}dt^\prime} +\int_{t_0}^{t}{\frac{1}{6}{x^3t^\prime}\exp{\left(\frac{x^3}{2t^\prime}-t^\prime\right)}dt^\prime}.
\end{split} 
\end{equation} 
Repeating the partial integration for the second and third terms of Eq. (15), we obtain
\begin{equation}
\int_{t_0}^{t}{\frac{1}{{F(t^\prime)}^2}dt^\prime}=\left[\frac{1}{3}{t^\prime}^3\exp{\left(\frac{x^3}{2t^\prime}-xt^\prime\right)}\right]_{t_0}^t
+{C^\prime(t}^4)+{o(t}^4).\ 
\end{equation}
Here, we note that the term ${o(t}^4)$ is small enough to be omitted because $t$ is bounded as $t\in[0,1]$. Using Eqs. (13), (14) and (16), we obtain the variance of $x(t)$ scales as: 
\begin{equation}
\langle x(t)^2 \rangle \simeq \kappa t^1+{C(t}^2)+\ {C^\prime(t}^3).\ \ 
\end{equation}
In the subsequent discussion, we shall demonstrate the intrinsic equivalence between the KPZ equation in Eq. (1) and the nonlinear Langevin equation in Eq. (8) using the above results with appropriate height function $h\left(x,t\right)$. Let us choose $h(x,t)$ as:
\begin{equation}
h\left(x,t\right):=\frac{3t^2x+x^3}{6t}.\ 
\end{equation}
 For the above choice of $h(x,t)$, we immediately obtain the following relations.
\begin{align}
\frac{\partial h\left(x,t\right)}{\partial t}&=\frac{1}{6}x\left(3-\frac{x^2}{t^2}\right),\\
\frac{\partial^2h\left(x,t\right)}{\partial x^2}&=\frac{x}{t},\\
\left(\frac{\partial h(x,t)}{\partial x}\right)^2&=\left(\frac{x^2+t^2}{2t}\right)^2.
\end{align}
Taking the ensemble averages of Eqs. (19)-(21), and using the time scaling relation in Eq. (17), we obtain the followings:
\begin{align}
\left\langle\frac{\partial h\left(x,t\right)}{\partial t}\right\rangle&=-\frac{\langle x^3 \rangle}{6t^2} \simeq -\frac{1}{6} {\langle x \rangle \frac{\kappa t^1+{C(t}^2)+ {C^\prime(t}^3)}{t^2}} = -\frac{\kappa\langle x \rangle}{6t} +C''(t^1),\\
\left\langle\frac{\partial^2h\left(x,t\right)}{\partial x^2}\right\rangle&=\frac{\langle x \rangle}{t},\\
\left\langle\left(\frac{\partial h\left(x,t\right)}{\partial x}\right)^2\right\rangle&=\frac{\langle x^4 \rangle + 2\langle x^2 \rangle t^2 +t^4}{4t^2}\simeq\frac{\kappa^2}{4}+\frac{1}{2}\kappa t+\frac{t^2}{4}. 
\end{align}
Here, we assumed the independence between $x$ and $x^2$ derived from Eq. (11) \footnote{Let us recall the expressions: $x(t)=\sqrt\kappa F\left(t\right)\int_{t_0}^{t}{\frac{\eta^{(1)}(t^\prime)\ }{F\left(t^\prime\right)}dt^\prime}$, $x(t)^2 = \kappa{F\left(t\right)}^2\int_{t_0}^{t}{\frac{1}{{F(t^\prime)}^2}dt^\prime}$, and $x(t)^3=\kappa\sqrt\kappa {F\left(t\right)}^3\int_{t_0}^{t}{\frac{\eta^{(2)}(t^\prime)\ }{{F(t^\prime)}^3}dt^\prime}$. Then, $x(t)$ and $x(t)^3$ are independent because of the independence of the noise term $\eta^{(1)}(t^\prime)$ and $\eta^{(2)}(t^\prime)$. Similarly, $x(t)^2$ and $x(t)^3$ are independent. Accordingly, $x(t)$ and $x(t)^2$ are assumed to be independent in the present stochastic calculus.}. Using Eqs. (25)-(27), we obtain the following relation:
\begin{equation}
\left\langle\frac{\partial h\left(x,t\right)}{\partial t}\right\rangle=-\frac{6}{\kappa}\left\langle\frac{\partial^2h\left(x,t\right)}{\partial x^2}\right\rangle+c(\kappa)\left\langle\left(\frac{\partial h\left(x,t\right)}{\partial x}\right)^2\right\rangle,
\end{equation}
where $c(\kappa)$ is $\kappa$-dependent constant term expressed as:
\begin{equation}
c(\kappa)=C^{\prime \prime}(t^1)/\left(\frac{\kappa^2}{4}+\frac{1}{2}\kappa t+\frac{t^2}{4}\right).\ 
\end{equation}
Noticing that the ensemble average of the white noise term in Eqs. (1) and (8) vanishes, i.e., $\langle \sqrt{\kappa} \eta(t)\rangle = 0$, the stochastic behavior individual orbit of $h(x,t)$ subjected to the noise fluctuation obeys the following KPZ-type nonlinear Langevin equation:
\begin{equation}
\frac{\partial h(x,t)}{\partial t}=\frac{6}{\kappa}\frac{\partial^2h(x,t)}{\partial x^2}-c(\kappa)\left(\frac{\partial h(x,t)}{\partial x}\right)^2+\sqrt\kappa\eta(t).\ 
\end{equation} 
In Eq. (27), we assumed the term $dB_s/ds$ as the white noise term $\eta(t)$, which is independent when the sampling time coordinate is changed from $s$ to $t$. Consequently, we observe that Eq. (27) is equivalent to KPZ equation in Eq. (1). Then, we obtain the following. 

\begin{theorem}
The ensemble of the height function $h(x,t)$ of the one-dimensional KPZ equation expressed by Eq. (1) behaves as the same law as that of $h(x,t)=(3t^2x+x^3)/6t$ of the modified SLE described by Eqs. (4) and (5) under the approximation that we omit the $o(t^4)$ term of $t\in[0,1]$.   
\end{theorem}
\begin{proof}
It follows from Eqs. (1), (18), (25) and the definition of the ensemble average.  
\end{proof}

To investigate the advantage of the Loewner equation-based description of KPZ equation, in the next section, we discuss the above results using Loewner entropy, which is the complexity-measure of the dynamics of the curve $\gamma_{[0,s]}$ generated by Loewner equation.

\subsection{KPZ universality class in terms of Loewner entropy}
In this subsection, we discuss the correspondence between the KPZ universality class (i.e., the observed $t^{1/3}$ scaling) and the Loewner entropy of the Loewner equation expressed by Eqs. (4) and (5). In accordance with Refs. [37,38], the Loewner entropy is defined as a Boltzmann-type entropy of the Loewner driving force $\eta_s:=dU_s/ds$, that is [37,38]:\footnote{In this article, the Loewner entropy $S_{Loew}$ is defined as the entropy of $s$-derivative of the Loewner driving function $U_s$.}:
\begin{equation}
S_{Loew}:=-\ln{p\left(\eta_s\right)}.\ 
\end{equation} 
To calculate the probability density function $p\left(\eta_s\right)$, we use the following relation:
\begin{equation}
\langle \eta_{s}^2 \rangle = \int{\eta_s^2p\left(\eta_s\right)}\eta_s.\
\end{equation}
Using Eq. (5), the variance $\langle {\eta_s}^2 \rangle$ is calculated as the following:
\begin{align}
\langle \eta_{s}^2 \rangle &= \left\langle\left(-\frac{x^2+t^2}{2t}-\sqrt\kappa\frac{2t}{x^2+t^2}\eta(t)\right)^2\right\rangle\\
&=\left\langle\left(\frac{x^2+t^2}{2t}\right)^2\right\rangle+2\sqrt{\kappa}\langle \eta(t) \rangle+\kappa\left\langle\left(\frac{2t}{x^2+t^2}\right)^2 \right\rangle \langle {\eta(t)}^2 \rangle\\
&=\left\langle\left(\frac{x^2+t^2}{2t}\right)^2+\kappa \left(\frac{2t}{x^2+t^2}\right)^2\right\rangle.
\end{align} 
From Eqs. (29) and (32), $p(\eta_s)$ is calculated as:
\begin{equation}
p\left(\eta_s\right)=\frac{1}{{\eta_s}^2}\frac{d\langle {\eta_s}^2 \rangle}{d \eta_s}\simeq\frac{2}{{\eta_s}^2} \langle {\eta_s} \rangle=\frac{2\left\langle\frac{x^2+t^2}{2t}\right\rangle}{\left(\frac{x^2+t^2}{2t}\right)^2+\kappa \left(\frac{2t}{x^2+t^2}\right)^2}. 
\end{equation}
From Eqs. (28) and (33), the Loewner entropy is expressed as: 
\begin{equation}
S_{Loew} =-\ln{p(\eta_s)}=-\ln{2\left\langle\frac{x^2+t^2}{2t}\right\rangle}+\ln{2\left(\frac{x^2+t^2}{2t}\right)^2}+\ln{\kappa \left(\frac{2t}{x^2+t^2}\right)^2}.  
\end{equation}
Accordingly, we obtain the following relation:
\begin{equation}
S_{Loew}\simeq-\ln{\frac{t}{\kappa}}.
\end{equation}
This condition is rewritten as the following relation. 
\begin{equation}
\exp{\left(-S_{Loew}\right)}=p(\eta_s)\propto t^{1},\ \ \ \\ \ t\in\left[0,1\right].\ 
\end{equation}
Thus, it is expected that the KPZ universality class is equivalent to the dynamics of the Loewner equation with the above condition on $S_{Loew}$. In other words, we obtain,

\begin{theorem}
For the dynamics of the KPZ equation in Eq. (1), the scaling relation expressed by Eqs. (2) and (3) with $\alpha=1/2$,\ $\beta=1/3$ and $z=3/2$ corresponds to $S_{Loew}\simeq-\ln{(t/\kappa)}$ when we regard the system size as $L \rightarrow x$. 
\end{theorem}
\begin{proof}
It follows from the arguments in Sec. 1 and 2, and the one-to-one correspondence of the curve and driving function of the Loewner equation.
\end{proof}
 
\section{Numerical simulation}
 The first numerical simulation was performed to test whether the width $W(L,t)$ of the height function $h(x,t)$ defined by Eq. (18) driven by the dynamics of $x$ described by Eq. (8) obeys the KPZ universality class in Eq. (3). We compute the time series of $x(t)$ and $h(x,t)$ by the discretization of Eq. (8) with Euler method as: $x_{n+1}=x_n-\tau(x_n/n\tau-{(({x_n}^2+\left(n\tau\right)^2)/n\tau)}^2-\sqrt{\kappa\tau}\eta_n$, where $\tau$ is the time step interval and $\eta_n$ is the white Gaussian noise with the variance 1.0 and mean 0. Similarly, $h(x,t)$ is discretized as $h\left(x_n,n\right)=(3\left(n\tau\right)^2x_n+{x_n}^3)/6\left(n\tau\right)$, and this is computed by successively substituting the obtained value of $x_n$. In the present study, we set $\tau=0.0001$, $N=10000$, and $\kappa=0.1$. Subsequently, the width function of $h\left(x_n,n\right)$ is calculated as $W(x_n,t):=  \sqrt{\langle h\left(x_n,n\tau\right)^2 \rangle - {\langle h\left(x_n,n\tau\right) \rangle}^2} $, where the ensemble average was taken over 200 realizations of the trajectories. Figure 1 shows the log-log plot of $t$ and $W(x_n,t)$ obtained from the above numerical settings. The scaling close to 1/3 exponent was observed in short time region, and 3/2 exponent was observed in the long time region. For this plot, because the system size $L$ is regarded as $L\sim t^3$, the obtained scaling is consistent with the KPZ universality class described in Eq. (3) ($\alpha=1/2$ and $\beta=1/3$). In addition, $t^{*}$ is approximated as $t^{*}\simeq 2.0\times10^{-2}$ for the plot of Fig.1. We note that the effective value of $\kappa$ is limited to the range of $0<\kappa\leq0.13$, and otherwise it results in the overflow of the variables in this numerical settings.\\ 
\begin{figure}[htbp]
\centerline{\includegraphics[width=9cm]{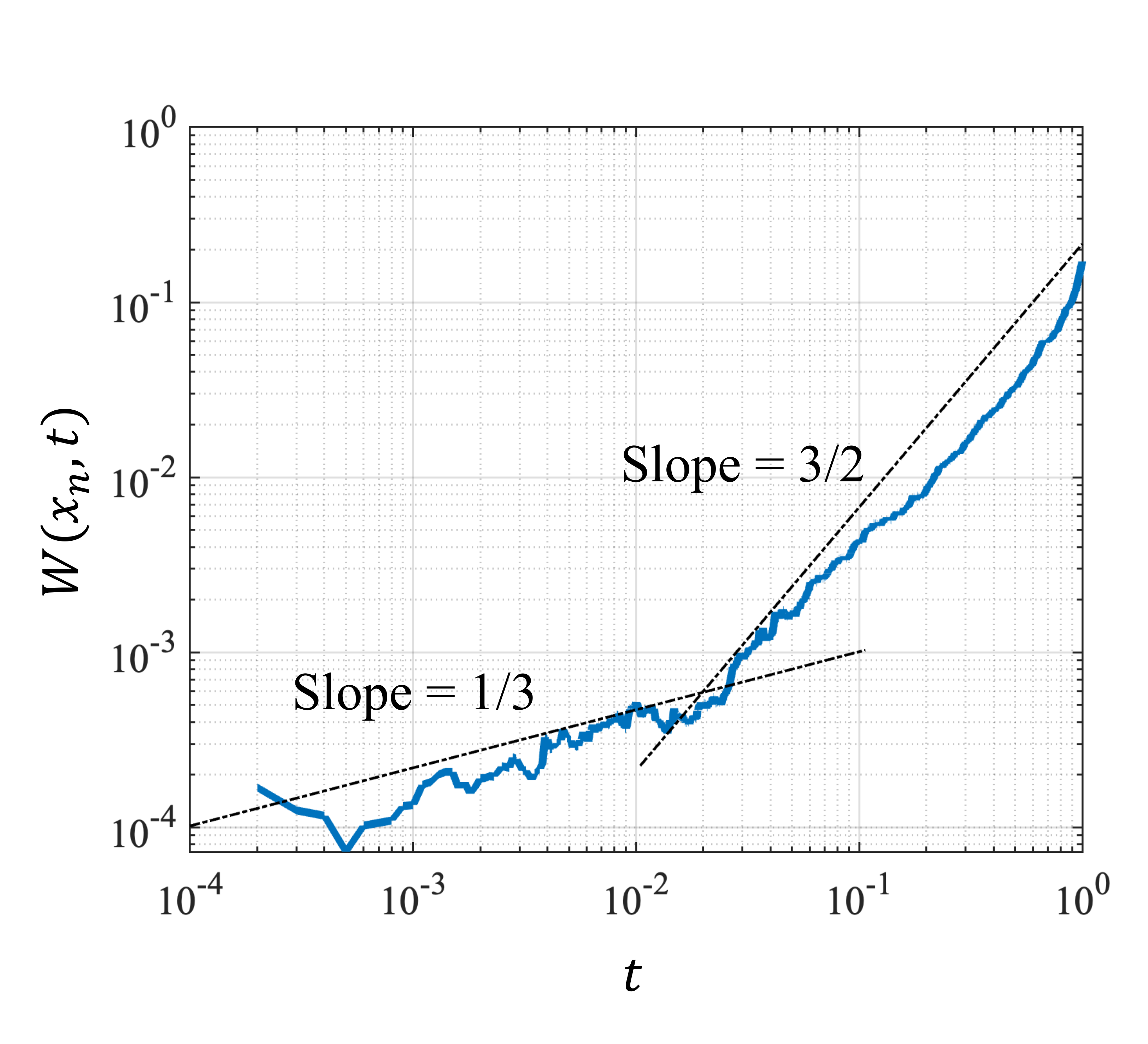}}
\caption{Log-log plot of $t$ and $W(x_n,t)$ obtained by the numerical simulation. The parameter settings are noted in the main text. The black solid lines show the scalings suggested by the theory of the KPZ universality class, i.e., $t^{1/3}$ in the short time region and $t^{3/2}$ in the long time region, respectively.}
\end{figure} 
  The second numerical simulation was performed to test the time dependence of the Loewner entropy. To verify the relations in Eqs. (35) and (36), the Loewner driving force $\eta_s$ was calculated by the zipper algorithm using the vertical slit map: $g_n(z)=\Delta U_{s_n} + \sqrt{(z-\Delta U_{s_n})+4\Delta s_n}$ [39]. We first compose the curve $\gamma_{[0,s]} = \{z_0(=0), z_1=x_1+i\tau, ..., z_n = x_n + i n\tau, ..., z_N = x_N+ i N \tau\}$, where $i=\sqrt{-1}$. Subsequently, we numerically obtained the sequence of the increments $\{\Delta U_{s_n}\}$ and $\{\Delta s_n\}$ of the Loewner driving function $U_s$ corresponding to the curve $\gamma_{[0,s]}$, and define the Loewner driving force $\eta_s(n)=\Delta U_{s_n}/\sqrt{\Delta s_n}$. Using the relation $p(\eta_s)\simeq (2/{\eta_s}^2) \langle \eta_s \rangle$ in Eq. (36), $p(\eta_s)=\exp{(-S_{Loew})}$ was numerically calculated. Figure 2 shows the log-log plot of $p(\eta_s)$ with $\kappa=0.1$. The $t^1$ scaling of $p(\eta_s)$ was observed in the whole time region. This result indicates the time-dependence of the probability distribution of the Loewner driving force $\eta_s$ and consistent with the analytical results in Sec. 2.2.  
\begin{figure}[htbp]
\centerline{\includegraphics[width=9cm]{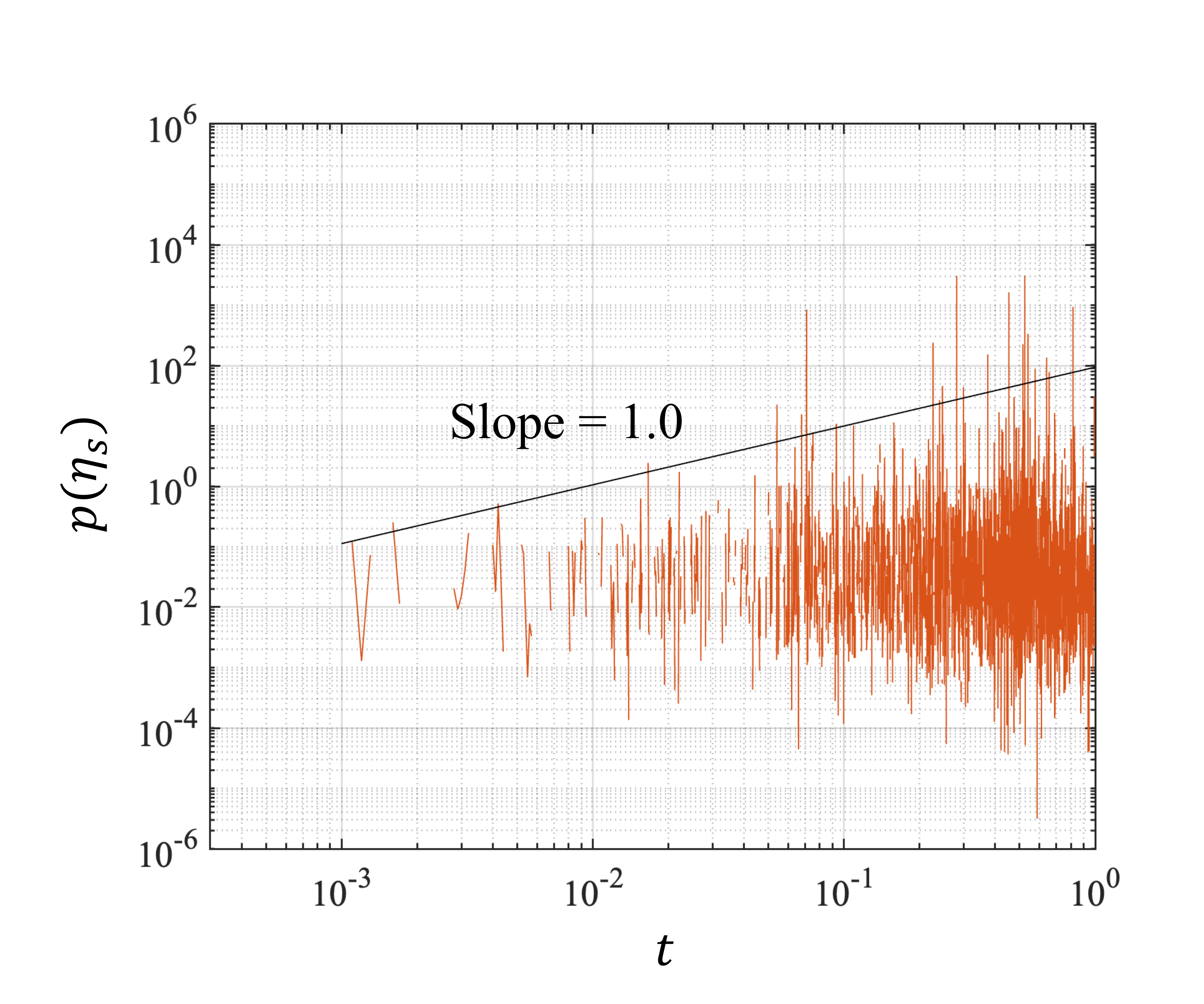}}
\caption{The scaling of the probability distribution of the Loewner driving force. The calculation method is described in the main text. The black solid line show the scaling suggested by Eq. (36), i.e., $\exp{\left(-S_{Loew}\right)}=p(\eta_s)\propto t^{1}$.}
\end{figure} 
\section{Discussion}
In the present study, the relationship between the 1D KPZ equation and SLE was investigated. By choosing the nonlinear stochastic process described by Eq. (5) as the driving function of the chordal Loewner evolution, the 1D KPZ-type equation was derived with the height function in Eq. (18). The analytical results were numerically verified, and we observed the scalings of the KPZ universality class in Eq. (3) from the plot of $t$ and $W(x_n, t)$. This will benefit the studies on the exact solution of the KPZ equation although it should be reminded that this analytical result includes the effect of the approximations. Further experimental studies are required for the verification of this theoretical scheme and parameters in the real non-equilibrium phenomena (e.g., the self-organization phenomena [40], curve-like geometries [41], the neuronal morphology or morphogenesis [42,43], and the other interfacial growth [44,45] in the experimental systems.) to clarify the physical meaning of the present result. In particular, the present model is restricted to time range of $t\in[0,1]$, the practical applications to the experimental situations requires the rescaling of the time coordinate. 
We also demonstrated that the analyses using Loewner entropy $S_{Loew}$ leads to the scaling relation of the 1D dynamics of KPZ equation described by Eqs. (35) and (36). Thus, the results in Eqs. (33)-(36) would suggest a certain universality class of the 1D nonlinear dynamics that is closely related to the KPZ universality class. It is worth mentioning that, the Loewner entropy defined by Eq. (28) is one of the conformal invariant of the curve $\gamma_{[0,s]}$ (thus, $x(t)$) in the sense that it is invariant under the conformal transformations $g_s$. Therefore, the present study suggest a perspective of the classification method of the nonlinear dynamics based on the conformally invariant entropy $S_{Loew}$ mapped on the complex plane.







\vskip2pc

\section{Conclusion}
We investigated the description of the 1D KPZ equation by SLE driven by a nonlinear stochastic process. The analytical and numerical results suggested that the dynamics of the1D KPZ equation with height function $h(x,t)=(3t^2x+x^3)/6t$ corresponds to those of the Loewner evolution with Loewner entropy $S_{Loew}\simeq-\ln{t/\kappa}$, and the KPZ universality class is valid also in this SLE-based description. These results will aid to the problem of the exact solution of the KPZ equation, while providing the novel understanding and classification method of the 1D nonlinear dynamics. Further experimental studies are required to verify the applicability of the present theoretical results to the real-world systems in non-equilibrium settings. 

\section*{Acknowledgment}
The main result of this study has been reported in JPS 2025 Spring Meeting of the Physical Society of Japan. The author sincerely thanks his family and friends who has morally supported the present research throughout.     



%



\let\doi\relax



\end{document}